\documentclass[journal]{IEEEtran}
\usepackage{romannum}
\usepackage[usenames,dvipsnames]{xcolor}
\usepackage{tkz-berge}
\usetikzlibrary{fit,shapes}                                     
\usepackage{calrsfs}
\DeclareMathAlphabet{\pazocal}{OMS}{zplm}{m}{n}
\usepackage{graphicx}
\usepackage{graphics} 
\usepackage{epsfig} 
\usepackage{mathptmx}
\usepackage{caption}
\usepackage{subcaption}
\usepackage{times} 
\usepackage{tikz}
\usetikzlibrary{positioning}
\usepackage{amsmath, amssymb}

\usepackage{amsthm}
\usepackage{amsfonts} 
\usepackage{setspace}
\usepackage{multirow}
\usepackage{textcomp}
\usepackage[english]{babel}
\usepackage{float} 
\usepackage{algorithmicx}
\usepackage[section]{algorithm}
\usepackage{algpascal}
\usepackage{algc}
\usepackage{algcompatible}
\usepackage{algpseudocode}
\let\oldReturn\Return
\renewcommand{\Return}{\State\oldReturn}
\usepackage{mathtools}
\usepackage{bm}
\usepackage{mathptmx}
\usepackage{times} 
\usepackage{mathtools, cuted}
\usepackage{textcomp}
\usepackage[english]{babel}
\usepackage{rotating}
\usepackage{pgfplots}
\pgfplotsset{compat=1.5}
\usepackage{siunitx}
\usepackage{pgfplotstable}
\usepackage{filecontents}

\newtheorem{theorem}{Theorem}
\newtheorem{lemma}[theorem]{Lemma}

\newtheorem{defn}[theorem]{Definition}
\newtheorem{prop}[theorem]{Proposition}
\newcommand{\Mb}{\pazocal{M}}
\newcommand{\Ma}{\pazocal{M}}

\newcommand{\sig}{\sigma_{e}(I, F)}

\newcommand{\EP}{\hspace*{\fill} {\boldmath $\square$ } \bigskip }

\graphicspath{{figures/}}
\numberwithin{theorem}{section}

\newcommand{\remove}[1]{}
\def \cN{{\cal N}}
\def \cC{{\cal C}}
\def \cM{{\Mb}}
\def \cMb{{{\Mb}(\beta)}}
\def \cS{{\cal S}}

\def \*{\star}
\def \10n{\!\!\!\!\!\!\!\!\!\!}
\def \te{{\tau_\epsilon}}
\def \mp {{M^\prime}}
\def \tm {{\tilde{M}}}

\makeatletter
\def\bbordermatrix#1{\begingroup \m@th
  \@tempdima 4.75\p@
  \setbox\z@\vbox{%
    \def\cr{\crcr\noalign{\kern2\p@\global\let\cr\endline}}%
    \ialign{$##$\hfil\kern2\p@\kern\@tempdima&\thinspace\hfil$##$\hfil
      &&\quad\hfil$##$\hfil\crcr
      \omit\strut\hfil\crcr\noalign{\kern-\baselineskip}%
      #1\crcr\omit\strut\cr}}%
  \setbox\tw@\vbox{\unvcopy\z@\global\setbox\@ne\lastbox}%
  \setbox\tw@\hbox{\unhbox\@ne\unskip\global\setbox\@ne\lastbox}%
  \setbox\tw@\hbox{$\kern\wd\@ne\kern-\@tempdima\left[\kern-\wd\@ne
    \global\setbox\@ne\vbox{\box\@ne\kern2\p@}%
    \vcenter{\kern-\ht\@ne\unvbox\z@\kern-\baselineskip}\,\right]$}%
  \null\;\vbox{\kern\ht\@ne\box\tw@}\endgroup}
\makeatother

\newcommand{\prob}{\rm Prob}

\ifCLASSINFOpdf

\else
\fi
\begin{document}
\title{\LARGE \bf Rapidly Mixing Markov Chain Monte Carlo Technique for Matching Problems with Global Utility Function}

\author{Shana~Moothedath,
        Prasanna~Chaporkar
        and~Madhu~N.~Belur
\thanks{The authors are in the Department of Electrical Engineering, Indian Institute of Technology Bombay, India. Email: $\lbrace$shana, chaporkar, belur$\rbrace$@ee.iitb.ac.in.}}


\maketitle

\begin{abstract}
This paper deals with a complete bipartite matching problem with the objective 
of finding an optimal matching that maximizes a certain generic predefined utility function on
the set of all matchings. 
After proving the NP-hardness of the problem using reduction from the 3-SAT problem, 
we propose a randomized algorithm based on Markov Chain Monte Carlo (MCMC) technique for 
solving this. 
We sample from Gibb's distribution and construct a reversible positive recurrent discrete time Markov chain (DTMC) that has the steady state distribution same as the Gibb's distribution. 
In one of our key contributions, we show that the constructed chain is `rapid mixing', i.e. 
the convergence time to reach within a specified distance to the desired distribution 
is polynomial in the problem size. 
The rapid mixing property is established by obtaining a lower bound on the conductance
of the DTMC graph and this result is of independent interest.
\end{abstract}


\IEEEpeerreviewmaketitle

\section{Introduction} \label{sec:intro}
Bipartite graphs arise in many applications. Matchings between 
two types of elements, like men/women, or jobs/machines are naturally
bipartite. Other areas include for example structural inconsistency detection in large electrical network \cite{ReiFel:03}, programmable logic arrays (PLAs) \cite{Hu:87} and electronic design automation \cite{WanChaChe:09}. Often, in a
bipartite graph, it is essential to find an optimal matching, one for which a certain suitable utility function has to be maximized.  This paper deals with this problem. However, the bipartite matching problem with global utility maximization turns out to be NP-hard. To this end, we propose a randomized algorithm using Markov Chain Monte Carlo technique for solving this. We use conductance \cite{Gur:00} for proving the rapid mixing of the chain. The conductance (for a graph with/without weights) gives an indication of how fast a random walk on a graph converges to a stationary distribution \cite{JerSin:89}. 

Conductance is one of many ways to quantify connectedness in a graph.
In the context of a Markov chain, well-connectedness results in, loosely
speaking and under suitable assumptions, a faster convergence to the stationary distribution. 
Network connectivity plays a key role in many applications. For example,
in circuits, one needs well-connectedness to have more 
reliability/redundancy against disconnectedness of some components. 
In traffic networks, high connectivity ensures more robustness against traffic-jams and delays. 
The approach followed in this paper for finding an optimal matching is by constructing
a Markov chain and then by ensuring that this chain is rapidly mixing, i.e. well-connected
in a suitable sense.  In what follows, we formulate the problem precisely and
then consider a few applications.

Consider a complete bipartite graph $G = ((V_1 \cup V_2),E)$. The graph is complete in a sense that
$(i,j) \in E$ for every $i\in V_1$ and $j\in V_2$. Let $|V_1| =m$ and $|V_2|=n$, and without loss
of generality $m \leqslant n$. A matching $M$ in $G$ is a collection of edges (subset of $E$) 
such that no two edges in the collection share the same endpoint, i.e. for any $(i,j)$ and $(u,v) \in M$, we have
$i \not= u$ and $j \not= v$. A matching $M \subseteq E$ is said to be perfect if for any $(i,j) \not\in M$,
$\{(i,j)\} \cup M$ is not a matching. 
Note that for any perfect matching $M$, $|M| = m$.
Let $\cN$ denote the set of all perfect matchings in $G$.
Now, consider a real valued function $U: \cN \to \Re$, where $\Re$ is the set of real numbers.
The function $U$ can be thought as assigning utility to each perfect matching. 
Our aim is to find the perfect matching $M^\*$ that maximizes the utility.
Specifically, we wish to solve the following optimization:
\[ M^\* = \arg\max_{\10n M \in \cN} U(M). \] 
We do not consider any specific structure on $U(\cdot)$.
However, we assume that given any perfect matching $M$,
$U(M)$ can be computed in time polynomial in $m$.

Note that the bipartite matchings is a very well studied problem on account of its
usefulness in modelling, among many others, scheduling and resource allocation problems. 
In these explorations, mostly some structure on $U(\cdot)$ is assumed. For example,
it is assumed that each edge $e$ in $E$ is associated with a non-negative real number called weight, say $w_e$.
Here, $U(M) = \sum_{e\in M} w_e$. 
Note that for a given $M$, $U(M)$ can be computed in $O(m)$.
In this settings, $M^\*$ is called {\it maximum weighted matching}.
The Hungarian algorithm can be used to obtain the maximum weighted matching in time complexity
$O(n^4)$ \cite{Kuh:10}. 
One other matching problem studied extensively is the {\it stable matching problem}.
Here, each node in $V_1$ ($V_2$, resp.) give preference for each node in $V_2$ ($V_1$, resp.).
A perfect matching is called stable if, broadly, there does not exist any pairing $(i,j)$
such that both $i$ and $j$ prefer each other more than the nodes they are currently matched to \cite{GalSha:62}. Here, note that we can define $U(M)$ to be 0 if $M$ is not a stable matching and 1 otherwise.
Note that for a given $M$, $U(M)$ can be computed in $O(m^2n^2)$. Also, $M^\*$ gives the stable matching. 
Thus, our problem here is a generalized version of well studied matching problems.
Note that the aforementioned well studied problems are not enough to model many 
resource allocation and scheduling problems. We give few examples to demonstrate this.

\subsection{Job Scheduling}
Consider a job scheduling problem where $m$ jobs need to be scheduled on $K$ machines.
Each job can be scheduled on any of the available machines.
Let $s_{ij}$ denote the service time of job $i$ on machine $j$.
Let $\cS$ denote the set of all jobs.
A job scheduling $\Delta$ is a partition of $\cS$ into at most $K$ subsets,
say $\cS_1^\Delta,\ldots,\cS_K^\Delta$, where
$\cS_j$ denote the set of jobs scheduled on machine $j$.
Denote by $T_j^\Delta$ the machine $j$'s make-span under $\Delta$ which is defined as follows:
\[ T_j^\Delta = \sum_{i \in \cS_j^\Delta} s_{ij}. \]
The system's make-span under $\Delta$, say $T^\Delta$, is defined as
$T^\Delta = \displaystyle \max_{1\leqslant j \leqslant K} T_j^\Delta$.
The aim is to find an optimal job scheduling $\Delta^\*$ such that
$T^{\Delta^\*} \leqslant T^\Delta$ for every $\Delta$. Now, we show that this problem can be
addressed in our framework. 

Construct a bipartite graph as follows: set $V_1=\cS$, i.e. $V_1$ is the set of all jobs.
Thus, $|V_1| = m$. Set $V_2$ has $mK$ nodes. Any perfect matching
$M$ can be mapped to a job scheduling as follows. 
Construct $\cS_j(M) = \{i \in \cS \ : \ (i,u) \in M \mbox{ and } j = \lceil u/m \rceil \}$.
Now, $U(M)$ can be computed as 
\[ U(M) = -\max_j \sum_{i \in \cS_j(M)} s_{ij}. \]
Note that $U(M)$ can be computed in $O(m)$. Also, note that $M^\*$ corresponds to $\Delta^\*$. Thus, the job scheduling problem can
be addressed in our framework.

\subsection{Graph Colouring}
Consider a graph $G_c = (V,E)$. Let $|V| = m$.
The graph colouring problem deals with assigning  colours to each vertex in such a way that
vertices $i$ and $j$ do not have the same colour  if $(i,j)\in E$. 
Here, we are interested in determining whether the given graph can be coloured with at most
$K$ colours \cite{Die:00}. 
The colouring problem is used, among many other things, for frequency planning in wireless networks.
Next we describe how this problem can be addressed in our framework.

Construct a bipartite graph as follows:
assign $V_1 = V$ and $V_2$ has $mK$ nodes.
Any perfect matching $M$ yields a graph colouring by
assigning colour $\lceil u/m \rceil$ to node $i$ if $(i,u) \in M$.
The utility function $U(M)$ equals $c$ if $M$ yields a valid colouring in the original graph and
equals $-c$ otherwise where $c$ is a positive real number.
Note that if $U(M^\*) = c$, then we can conclude that the graph can
be coloured with at most $K$ colours. Also, note that for any given
matching $M$, its utility can be obtained in $O(m^2)$. Thus, 
the proposed framework can address the colouring problem.

\subsection{Multiple Knapsack Problem}
The multiple knapsack problem is a generalization of the single knapsack problem.
Let there be $m$ items and $K$ knapsacks. The volume of the $i^{\rm th}$ item is $c_i > 0$
and the volume of the  $j^{\rm th}$ knapsack  is $C_i$.
Denote by $r_{ij} > 0$ the reward we obtain if an item $i$ is put in the $j^{\rm th}$ knapsack.
The aim is to assign items to knapsacks. Let $\cS_j^\Delta$ denote the set of items assigned to
knapsack $j$ under assignment policy $\Delta$. 
The assignment $\Delta$ is said to be feasible if $\sum_{i\in\cS_j^\Delta} c_i \leqslant C_j$ for every $j$.
For any feasible $\Delta$, the reward $R^\Delta = \sum_{j} \sum_{i\in \cS_j^\Delta} r_{ij}$.
Our aim is to find a feasible policy $\Delta^\*$ such that $R^{\Delta^\*} \geqslant R^\Delta$ 
for every feasible $\Delta$. Next we describe how this problem can be addressed in our framework.

Consider a bipartite graph with $V_1$ as the set of all items and $V_2$ is the set of $mK$ nodes.
A perfect matching $M$ in this bipartite graph is mapped to items assignment to the knapsack as follows:
item $i$ is assigned to the $\lceil j/m \rceil$ if $(i,j) \in M$. Define $\cS_j(M)$ to be the set of 
all items $i$ assigned to the $j^{\rm th}$ knapsack under matching $M$. Function $U(M) = -\kappa$ if the assignment
is not feasible, else $U(M) =  \sum_{j} \sum_{i\in \cS_j(M)} r_{ij}$. Note that $U(M)$ can be computed in
$O(m^2)$. Now, observe that $M^\*$ corresponds to $\Delta^\*$.

The examples above demonstrate usefulness of the framework we consider here.
Our aim is to design efficient algorithms to obtain $M^\*$. 
Unfortunately, the problem of finding optimal matching is NP-hard.
We prove the hardness in the next section.
Since the polynomial time algorithms for finding $M^\*$ may not exist,
we consider a randomized algorithm based on Markov Chain Monte Carlo \cite{AndDefDouJor:03}.
Key idea here is to sample efficiently from Gibb's distribution $\exp\{\beta U(M)\}/C$, 
where the partition function, $C =  \sum_{M\in \cN} e^{\{\beta U(M)\}}$ for $\beta > 0$. Note that the
distribution concentrates at the optimal perfect matchings as $\beta \to \infty$.
In MCMC, we construct  a reversible positive recurrent discrete time Markov chain (DTMC)
that has the steady state distribution same as the required Gibb's distribution.
Making the sampling from the required distribution efficient is equivalent to ensuring
that the MCMC Markov chain is rapid mixing. Broadly, the rapid mixing implies that after
running DTMC for steps polynomial in $n$, observed empirical distribution is very close to
the steady state distribution.
In our main result, we establish rapid mixing of the constructed DTMC. 
Our key contributions are summarized below.

\noindent
$\bullet$ We prove the NP-hardness of the problem.

\noindent
$\bullet$ We propose MCMC based randomized algorithm to find~$M^\*$.

\noindent
$\bullet$ We show that the constructed DTMC is rapid mixing for any given value of $\beta > 0$.

\noindent
$\bullet$ The bound on the conductance of the graph corresponding to the constructed DTMC may be
of independent interest.  \\

\noindent
 The organization of this paper is as follows: the complexity of the problem, that is, the NP-hardness is discussed in Section \ref{sec:comp}. 
The proposed algorithm for solving the problem at hand, the construction of the Markov chain and its transitions are detailed in Section \ref{sec:markov}. 
Section \ref{sec:rapid} deals with details of rapid mixing Markov chains and the tools used in this paper for establishing rapid mixing property. 
The main result of this paper, that is, the rapid mixing property of the Markov chain associated with the concerned problem is given in Section \ref{sec:mainresult}.  Section \ref{sec:conclu} contains concluding remarks and future directions. 
Proofs of few preliminary results is given in the Appendix.

\section{Complexity of the Problem}\label{sec:comp}
In this section, we discuss the computational complexity of the problem under consideration.
Note that we have a combinatorial optimization problem with no real structure apparent on the utility function.
The total number of perfect matchings in the bipartite graph is $O(n^m)$.
Note that the total number of matchings may not be polynomial in $m$ (or $n$) if $n$ is $O(m)$ which is the case in
the examples we described. Thus, the exhaustive search may not be computationally feasible for the cases in which
$n$ is increasing with $m$.
Next, we formally show that the allocation problem with global utility is NP-hard.

To show the problem is NP-hard, we use reduction from three-conjunctive normal form satisfiability 
(aka 3-CNF-SAT) problem.
We first describe 3-CNF-SAT problem for the sake of completeness.
This problem deals with determining satisfiability of a boolean formula that
involves $m$ variables and $K$ clauses, where $K$ is $O(m)$.
There are three boolean operators used, AND (denoted as `$\wedge$'), OR (denoted as `$+$') and complementation.
The clauses are joined with AND operators and variables in each clause are joined by OR operator.
Each clause has three variables either as it is or complemented.  
For example $(y_{1} + y_{2} +\bar{y}_{3})\wedge(y_{2} + y_{4} +\bar{y}_{7})\wedge (\bar{y}_{1} + y_{4} +\bar{y}_{8})$ is a 3-CNF. A 3-CNF is said to be satisfiable, if there exists assignment of binary value to each variable
such that the formula evaluates to 1. The 3-CNF-SAT is a known NP-complete problem \cite{CorLeiRivSte:01}.
Now, we prove that our problem is also NP-hard.

\begin{theorem} \label{th:npcomplete}
The allocation problem with global utility maximization is NP-hard.
\end{theorem} 
\begin{proof}
In order to prove the problem is NP-hard we use reduction from 3-CNF-SAT.
A 3-CNF satisfiability problem is represented as
\begin{align*}
V = \lbrace y_{1}, y_{2}, \cdots, y_{m}  \rbrace,
\end{align*}
and a 3-CNF boolean formula $f(V)$.
To find an analogy between the 3-CNF-SAT problem and the allocation problem with global utility maximization, we consider the following. Consider a bipartite graph $G((V_1 \cup V_2), E)$ with $V_1 = V$ number of nodes on one side and $|V_2| = 2m$ number of nodes on the other side as shown in Figure \ref{fig:bipartite}. 

\definecolor{myblue}{RGB}{80,80,160}
\definecolor{mygreen}{RGB}{80,160,80}
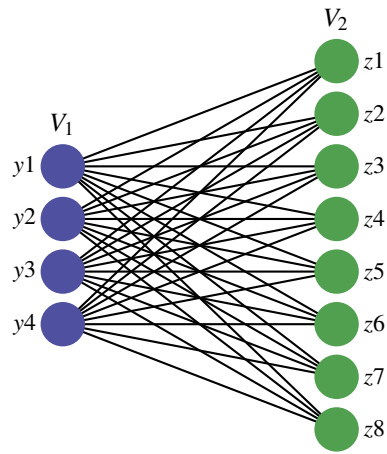
\begin{figure} 
\begin{center}
\begin{tikzpicture} [scale = 0.28]
       \draw [thick] (-7,-5)  --   (6,0);
       \draw [thick] (-7,-5)  --   (6,-2.5);
       \draw [thick] (-7,-5)  --   (6,-5);
       \draw [thick] (-7,-5)  --   (6,-7.5);
       \draw [thick] (-7,-5)  --   (6, -10);
       \draw [thick] (-7,-5)  --   (6, -12.5); 
       \draw [thick] (-7,-5)  --   (6, -15);
       \draw [thick] (-7,-5)  --   (6, -17.5);
       
       \draw [thick] (-7,-7.5)  --   (6,0);
       \draw [thick] (-7,-7.5)   --   (6,-2.5);
       \draw [thick] (-7,-7.5)  --   (6,-5);
       \draw [thick] (-7,-7.5)  --   (6,-7.5);
       \draw [thick] (-7,-7.5)   --   (6, -10);
       \draw [thick] (-7,-7.5)  --   (6, -12.5);
       \draw [thick] (-7,-7.5)  --   (6, -15);
       \draw [thick] (-7,-7.5)   --   (6, -17.5);
       
       \draw [thick] (-7,-10.0)  --   (6,0);
       \draw [thick] (-7,-10.0)  --   (6,-2.5);
       \draw [thick] (-7,-10.0)   --   (6,-5);
       \draw [thick] (-7,-10.0)  --   (6,-7.5);
       \draw [thick] (-7,-10.0)  --   (6, -10);
       \draw [thick] (-7,-10.0)   --   (6, -12.5);
       \draw [thick] (-7,-10.0)  --   (6, -15);
       \draw [thick] (-7,-10.0)  --   (6, -17.5);
       
       \draw [thick] (-7,-12.5)  --   (6,0);
       \draw [thick] (-7,-12.5)  --   (6,-2.5);
       \draw [thick] (-7,-12.5)   --   (6,-5);
       \draw [thick] (-7,-12.5)  --   (6,-7.5);
       \draw [thick] (-7,-12.5)  --   (6, -10);
       \draw [thick] (-7,-12.5)   --   (6, -12.5);
       \draw [thick] (-7,-12.5)  --   (6, -15);
       \draw [thick] (-7,-12.5)  --   (6, -17.5);     
          
          \node at (-8.8,-5) {\small $y1$};
          \node at (-8.8,-7.5) {\small $y2$};
          \node at (-8.8,-10.0) {\small $y3$};
          \node at (-8.8,-12.5) {\small $y4$};
          \fill[myblue] (-7,-5) circle (30.0 pt);
          \fill[myblue] (-7,-7.5) circle (30.0 pt);
          \fill[myblue] (-7,-10) circle (30.0 pt);
          \fill[myblue] (-7,-12.5) circle (30.0 pt);
          \node at (7.8,0) {\small $z1$};
          \node at (7.8,-2.5) {\small $z2$};
          \node at (7.8,-5) {\small $z3$};
          \node at (7.8,-7.5) {\small $z4$};
          \node at (7.8,-10) {\small $z5$};
          \node at (7.8,-12.5) {\small $z6$};
          \node at (7.8,-15) {\small $z7$};
          \node at (7.8,-17.5) {\small $z8$};         
          \fill[mygreen] (6,0) circle (30.0 pt); 
          \fill[mygreen] (6,-2.5) circle (30.0 pt);
          \fill[mygreen] (6,-5) circle (30.0 pt);
          \fill[mygreen] (6,-7.5) circle (30.0 pt); 
          \fill[mygreen] (6,-10) circle (30.0 pt);
          \fill[mygreen] (6,-12.5) circle (30.0 pt);
          \fill[mygreen] (6,-15) circle (30.0 pt);
          \fill[mygreen] (6,-17.5) circle (30.0 pt);
          \node at (-7.0,-3) {$V_1$};
          \node at (6,2) {$V_2$};
\end{tikzpicture}
\caption{The bipartite graph representation of 3-CNF satisfiability problem}\label{fig:bipartite}
\end{center}
\end{figure}
Out of the $2m$ nodes in the set $V_2$, first $m$ nodes $\{z_0,\ldots,z_{m-1}\}$ correspond to bit 0 
and the latter $m$ nodes $\{z_m,\ldots,z_{2m-1}\}$ correspond to bit 1. 
The graph $G$ shown in Figure \ref{fig:bipartite} is said to be complete since there exists edges from all nodes in set $V_1$ to all nodes in set $V_2$. 
Thus the 3-CNF satisfiability problem can be considered as a matching problem, where the allocation of each variable node $y_{i}$ to a node in $V_2$ can be associated with either 0 or 1  and this value is the value of that variable. To be more clear,  let $g(y_{i})$ denote the value of vertex in $V_2$ to which $y_{i}$ is matched. It can be either 0 or 1. The utility function $U$ for this is defined as 
\begin{equation}
U(M) = f(\bar{g}(\bar{y})) \in \lbrace 0, 1 \rbrace.
\end{equation}
The problem definition is find $M^{\*}$ such that
\begin{equation*}
M^\* = \arg\max_{\10n M \in \cN} U(M). 
\end{equation*}
 Here $U_{\mathrm{min}} = 0 $ and $ U_{\mathrm{max}} = 1$ and if $U(M^{\*}) = 1$, then $f$ is satisfiable, otherwise $f$ is not satisfiable. Thus the 3-CNF satisfiability problem can be reduced to a matching problem with global utility maximization in polynomial time. This proves the required.
\end{proof}

\noindent {\bf Remark}:
 The utility maximization with global utility is an NP-hard problem 
 even when the utility function is bounded, i.e. $|U(M)| \leqslant c$ for some $c > 0$.
 
On account of the NP hard nature of the problem, we can not have polynomial time algorithm to find
an optimal matching unless P=NP. There are two main approaches used to solve these problems:
approximation algorithm and randomized algorithms. In this paper, we focus on the randomized algorithm.
In the next section, we discuss our proposed algorithm.

\section{Proposed Randomized Method}\label{sec:markov}
We first present our randomized algorithm and then state some preliminary results.
\subsection{Proposed Algorithm}
The pseudo code for the algorithm is presented in Algorithm~\ref{alg:mcmc}.
We start with a random perfect matching from $\cN$, say $M_0$.
We run the algorithm for $T$ steps. In step $t$,  
we randomly choose a perfect matching $M_t \in \cN$ where the choice distribution
depends on $M_{t-1}$.  
Select $y \in V_1$ and $z \in V_2$ uniformly at random. Note that since we assume that $m \leqslant n$, in any perfect matching
node $y$ is always matched to some node in $V_2$. As for $z$, there are three possibilities:
(a)~$z$ is matched to $y$, i.e. the edge $(y, z)$ is already a part of the matching $M_{t-1}$.
(b)~$z$ is not matched, and $(y,z_1)$ is  a part of the matching $M_{t-1}$ for some $z_1 \in V_2$.
(c)~$z$ is matched to $y_1$, i.e. $(y,z_1)$ and $(y_1,z)$ both belong to $M_{t-1}$.
In each of the possible cases, we do the following:
In case (a), we do nothing and retain the same matching as before. Thus, $M_t = M_{t-1}$ in this case.
In case (b), we construct a perfect matching $\tilde{M}_{t}$ by removing edge $(y,z_1)$ and then
adding the  edge $(y,z)$ to  $M_{t-1}$. Similarly in case (c), we construct $\tilde{M}_{t}$ by removing
edges $(y,z_1)$ and $(y_1,z)$, and adding edges $(y,z)$ and $(y_1,z_1)$ to $M_{t-1}$. Once
the matching $\tilde{M}_{t}$ is constructed, we compute $U(\tilde{M}_{t})$. If $U(\tilde{M}_{t}) \geqslant U(M_{t-1})$,
then we let $M_t = \tilde{M}_{t}$; otherwise we only accept $\tilde{M}_{t}$ as a new choice with probability
$\exp\{\beta (U(\tilde{M}_{t})-U(M_{t-1}))\}$.
Note that in each iteration of the proposed algorithm, we randomly pick a matching from the neighborhood 
of the current one, and propose to use it. If the chosen matching has equal or more utility than that of the current
one, then we accept it as a new one, else we accept it only probabilistically with probability depending on the utility of
the proposed and the current matching. Specifically, closer the utility of the proposed matching to that of the existing one,
higher is the probability of accepting the proposed matching.
Next we prove some preliminary results.

\begin{algorithm}[t]
  \caption{Pseudo code for the proposed algorithm 
  \label{alg:mcmc}}
  \begin{algorithmic}
\State \textit {\bf Input:} Bipartite graph $G((V_1 \cup V_2, E)$, where $|V_1| = m, |V_2| = n$
\end{algorithmic}
  \begin{algorithmic}[1]
  \State  Initialize $V_1 = \{1,\cdots,m\}$ and $V_2 = \{1,\cdots,n\}$, step $t=0$ \label{step:initialize}
  \State Start with a random complete matching, say $M_0$
\While {$t \leqslant T$} \label{step:utilitycheck} 
    \State Select nodes $y \in V_1$ and $z \in V_2$   uniformly at random
    \If {$(y,z) \in M_t$}
	    \State  $\tilde{M}_{t+1} \leftarrow M_t$
	\ElsIf {$(y,z_1)\in M_t$ and $z$ is unmatched}
		\State $\tilde{M}_{t+1} \leftarrow (M_t \cup \{(y,z)\}) \setminus\{(y,z_1)\}$
	\ElsIf {$(y,z_1)\in M_t$ and $(y_1,z) \in M_t$}
		\State $\tilde{M}_{t+1} \leftarrow (M_t \cup \{(y,z),(y_1,z_1)\}) \setminus\{(y,z_1),(y_1,z)\}$
	\EndIf
	
\State Calculate $U(\tilde{M}_{t+1})$ 
\State $p = \min\{1,e^{\beta(U(\tilde{M}_{t+1})-U(M_t))}\}$
\State $M_{t+1} \leftarrow \tilde{M}_{t+1}$ w.p. $p$ and $M_{t+1} \leftarrow M_t$ otherwise
\State $t \leftarrow t+1$
\EndWhile \label{step:end}
\end{algorithmic}
\begin{algorithmic}
\State \textit{\bf Output:} Matching $M_T$
  \end{algorithmic}
\end{algorithm}

\subsection{Preliminary Results}
Fix parameter $\beta < \infty$. Let $X_t^\beta$ be a random variable that denotes the matching used by the proposed 
algorithm in $t^{\rm th}$ iteration for the given $\beta$. 
Consider the discrete time stochastic process $\cMb=\{X_t^\beta\}_{t\geqslant 0}$.
We prove the following two main results.

\begin{lemma} \label{lem:DTMC}
	Fix any $\beta < \infty$. The process $\cMb$ is a Discrete Time Markov Chain (DTMC) on $\cN$. 
	The DTMC is irreducible and aperiodic.
\end{lemma}
The proof is presented in the Appendix.
Let $P(\beta) = [P_{M^\prime M}(\beta)]_{M,M^\prime \in \cN}$ denote the transition probability matrix
on the DTMC for a given $\beta$.
Note that since $\cN$ has finitely many elements, the above result states that the DTMC is
positive recurrent and hence admits the stationary distribution, say $\bm{\pi}(\beta) = [\pi_M(\beta)]_{M\in\cN}$.
Now, we characterize the steady state distribution.

\begin{lemma} \label{lem:reversible}
	Fix any $\beta < \infty$. The DTMC $\cMb$ is time reversible and for every $M \in \cN$
	\[\pi_M(\beta) = \frac{\exp\{\beta U(M)\}}{\sum_{M^\prime\in \cN}\exp\{\beta U(M^\prime)\}}~.\]
\end{lemma}
The proof is presented in the Appendix. Note that the steady state distribution is what we desired as
it concentrates on points $M \in \cN$ such that $U(M) = U(M^\*)$ as $\beta \rightarrow \infty$. Also,
since $P_{M^\prime M}^{(t)}(\beta) \to \pi_M(\beta)$ as $t\rightarrow \infty$ for every $M^\prime$ and $M \in \cN$,
it should be enough to run the algorithm for some ``large enough'' steps $T$ to be able to ``closely'' sample
from the distribution $\bm{\pi}(\beta)$. However, determining $T$ for the required sampling accuracy
is challenging. Next we discuss this in more details.

\section{Rapid Mixing Markov Chain}\label{sec:rapid}
In probability theory, the mixing time of a Markov chain is the minimum time within which the Markov chain is 
``sufficiently close'' to its steady state distribution for all initial conditions. It gives a measure of how large $T$ should be so that the state of the chain $\cMb$ is close enough to the desired stationary distribution. 
Here onwards, we fix $\beta \in [0,\infty)$, and omit it from the notations for brevity. 
For DTMC $\cM$, define how far the $t$ step transition probability from the steady state distribution as 
{\it total variation distance} 
\begin{align} \label{eq:TV}
d(t) = \frac{1}{2}\max_{M^\prime \in \cN}\sum_{M \in \cN} |P_{M^\prime M}^{(t)} - \pi_M|.\end{align} 
The total variation distance $d(t)$ is a monotone decreasing function of $t$.
Thus, if $d(t) \leqslant \epsilon$, then sampling anytime after $t$ steps ensures that the sampled distribution is
at most $\epsilon$ away from the desired distribution $\bm{\pi}$ irrespective of the initial state. 
Thus, the $\epsilon$ mixing time, say $\tau_\epsilon$, is defined as follows for any $\epsilon>0$:
\[\te = \min\{t\ : \ d(t) \leqslant \epsilon  \} .   \]
A Markov chain is said to be {\it rapidly mixing} if the mixing time increases as polynomial in $n$ and
logarithmic in $1/\epsilon$. 
\begin{defn} \label{def:rapid}
	A Markov chain is said to be rapidly mixing if $\te$ is $O(poly(n)+\log(1/\epsilon))$ for every $\epsilon > 0$.
\end{defn}

Our aim is to show that $\cM$ is rapidly mixing. Since $\bm{\pi}$ is a left eigenvector of the transition
probability matrix $P$, from (\ref{eq:TV}) it is can be seen that to establish rapid mixing 
the spectral properties of $P$ have to be analyzed \cite{DiaStr:91}.
Note that the Markov chain associated with the problem is a reversible Markov chain. 
The standard tools that are used for characterizing the rapid mixing property of a reversible chain are conductance, 
canonical paths, coupling and path coupling. For a detailed description of these tools see \cite{Gur:00}. 
In order to prove the rapid mixing property of the chain $\cM$, 
we use conductance of the chain and the concept of canonical paths. A brief description of these are given below.

\subsection{Conductance}
Broadly, conductance of a graph is a measure of its connectedness.
Higher value of conductance typically imply that a large number
of edges will have to be removed for dividing the graph into multiple components. Similarly, the conductance of Markov chain graph
refers to ease with which various states can be traversed in the chain.
Specifically, conductance provides a direct way of bounding the
spectral gap of the transition probability matrix $P$ of $\cM$
using a geometric parameter based on a structural property of the
underlying weighted graph corresponding to $\cM$. Here, the
spectral gap refers to the distance between the unit circle 
and the second largest eigenvalue (by magnitude) of $P$.
Now, we formally define conductance of a Markov chain.
Let $w_{\mp M} = \pi_\mp P_{\mp M} = \pi_M P_{M\mp}$ denote
the weight on edge $(\mp,M)$.  
Note that weights are symmetric on account of time reversibility.
Moreover, the weight is zero if
$P_{\mp M} =0$.
For a subset $S$ of $\cN$, define $\pi_S = \sum_{M \in S} \pi_M$, i.e.
$\pi_S$ is the steady state probability of being in set $S$.
Let $S^c = \cN \setminus S$. Define,
\[ Q(S,S^c) = \sum_{M\in S, \mp\in S^c} \pi_M P_{M\mp,}\]
i.e. $Q(S,S^c)$ denote the flow from set $S$ to $S^c$ at steady state.
The conductance of $\cM$ is defined as
\begin{equation} \label{eq:conductance}
\Phi = \Phi(\Mb) := \displaystyle \min_{\substack{S \subseteq \cN \\ 0 < \pi_{S} \leqslant 1/2}} \frac{ Q(S, S^{c})} {\pi_{S}}.
\end{equation}
Note that conductance captures ability of Markov chain to move out of low probability sets.
The mixing time can be high if a Markov chain gets stuck in such sets and thereby limiting its ability to explore the state space.
The quantity $Q(S,S^c)$ captures the steady state probability of moving from set $S$ to $S^c$.
In the definition, we focus on probabilistically non-dominant sets $S$ that have probability less than
equal to 1/2. For fast mixing, we need the chain not to be stuck for a long time in low probability sets.
Thus, a higher value of conductance imply a shorter mixing time.
 
 \remove{
Rapid mixing of a Markov chain amounts to bounding the magnitude of the second farthest from origin eigenvalue of the probability transition matrix $P$. Conductance is a direct way of bounding the spectral gap of matrix $P$, using a geometric parameter which is based 
on a structural property of the underlying weighted 
graph. The underlying graph of a Markov chain is the graph that has the state space as its vertex set, each pair of state $i, j$  has a weight associated with its edge given by $w_{ij} = \pi_{i}P_{i, j} = \pi_{j}P_{j, i}$. Thus whenever $P_{i, j} > 0$, there exists an edge corresponding to those vertices. This characterisation is related to \cite{Alo:86} and \cite{AloMil:85} on eigenvalues and expander graphs. The conductance of a Markov chain $\Mb$ is defined as 
\begin{equation} 
\Phi = \Phi(\Mb) := \displaystyle \min_{\substack{S \subset \Omega \\ 0 < \pi_{S} \leqslant 1/2}} \frac{ Q(S, S^{c})} {\pi_{S}},
\end{equation}\label{eq:conductance}
where $S^{c} = \Omega \setminus S$, $Q(i, j) = \pi_{i}P_{i, j} = \pi_{j}P_{j, i}$, $\pi_{S}$ is the probability density under the stationary distribution $\pi$ of $\Mb$, and $Q(S, S^{c})$ is the sum of $Q(i, j)$ over all $(i, j) \in S \times S^{c}$. $\Phi$ measures the ability of $\Mb$ to escape from any small region of the state space, and hence to make rapid progress to the stationary distribution.}

\subsection{Canonical Paths}
Concept of canonical paths is used to quantify conductance of the Markov chain graph.
Recall that the conductance is a measure of graph connectedness. To ensure connectedness,
we need to ensure that there do not exist bottlenecks in the graph. 
Thus, the slow mixing is characterized
by bottlenecks, which is a set of edges whose total weight is small and their removal disconnects the state space into two exponentially large sets. Thus, presence of a bottleneck in the chain, results in slow mixing of the chain, as it takes exponential time to move from one side to the other. On the other hand, absence of a bottleneck ensures that no transition of the chain is used by too many paths. Construction of canonical paths aids in quantifying the number of paths passing through an edge.
Canonical path between any initial matching $I$ and final matching $F$ is a specially constructed unique positive probability 
simple path from $I$ to $F$ in the underlying Markov chain graph.  Using these paths, we obtain a lower bound on the graph
conductance. We describe the construction in detail. First, let $A \oplus B$ denote the symmetric difference between sets
$A$ and $B$.

\remove{
Let $H(\Mb)$ be the underlying graph associated with an ergodic Markov chain $\Mb$ on a finite set $\Omega$. Let $e$ be an edge between an ordered pair $(i, j)$ of weight $Q(e) = Q(i, j) = \pi_{i}P_{i, j}$ whenever $P_{i, j} > 0$. A set of canonical paths for $\Mb$ is a set $\Gamma$ of simple paths $\gamma_{ij}$ in the graph $H(\Mb)$, one between each ordered pair $(i, j)$ of distinct vertices. Thus to have a lower bound on the conductance of $\Mb$, we should have a set of canonical paths that do not overload any transition of $\Mb$. The pseudo code for constructing the canonical path is given in Algorithm~\ref{alg:canonicalpath}. }

\begin{algorithm}[t]
  \caption{Pseudo code for constructing canonical path 
  \label{alg:canonicalpath}}
   \begin{algorithmic}
\State \textit {\bf Input:} States $I \in \cN$ and  $F \in \cN$.
\State \textit {\bf Initialize:} $k=1$, $\tm_1 \leftarrow I$ and $J \leftarrow I \oplus F$
\end{algorithmic}
  \begin{algorithmic}[1]
  \For {$k = 1,2,\cdots,m$}
  \If {$y_k \in J$ }
  \State Let $(y_k,z) \in \tm_k \cap J$ and $(y_k,z^\prime) \in J \cap F$
  \If {there exist $(y^\prime,z^\prime)\in \tm_k$}
	  \State $\tm_{k+1} \leftarrow (\tm_k\cup\{(y_k,z^\prime),(y^\prime,z)\})\setminus \{(y_k,z),(y^\prime,z^\prime)\} $ \label{step:a2_5}
  \Else
	  \State $\tm_{k+1} \leftarrow (\tm_k \cup \{(y_k,z^\prime)\}) \setminus \{(y_k,z)\}$ \label{step:a2_7}
  \EndIf
  \Else 
  \State $\tm_{k+1} \leftarrow \tm_k$
  
  \EndIf
  \State $k \leftarrow k+1$
  \State $J \leftarrow \tm_k \oplus F $  
  \EndFor
\end{algorithmic}
\begin{algorithmic}
\State \textit{\bf Output:} Canonical path from $I$ to $F$ as $I=\tm_1,\tm_2,\ldots,\tm_k=F$
  \end{algorithmic}
\end{algorithm}
Fix any two distinct matchings $I$ and $F$ in $\cN$.
The pseudo code for constructing the canonical path from $I$ to $F$ is given in Algorithm \ref{alg:canonicalpath}. Let us say that the nodes in set $V_1$ is indexed as $V_1 = \{y_1, y_2, \cdots, y_m\}$. Given $I$ and $F$ we start with the first node in the set $V_1$, i.e. $y_1$. Note that, for $J = I \oplus F$, $\mbox{deg}_{y_1}(J) = \{0, 2\}$. If $\mbox{deg}_{y_1}(J) = 0$, then the corresponding edge is present in both $I$ and $F$. So, the next state is the current state itself. 
However, if the $\mbox{deg}_{y_1}(J) =  2$ then the corresponding edges are $(y_1,z)$ and $(y_1,z^\prime)$ in $I$ and $F$ respectively. 
We want to connect $y_1$ to $z^\prime$. If $z^\prime$ is not matched in $I$, then we simply break $(y_1,z)$ and
form $(y_1,z^\prime)$ and this gives us matching $\tm_2$ (see Step~\ref{step:a2_7}).
If $z^\prime$ is already matched to some $y^\prime$ in $I$, then we remove edges $(y_1,z)$ and $(y^\prime,z^\prime)$ and add
edges $(y_1,z^\prime)$ and $(y^\prime,z)$ to get $\tm_2$ (see Step~\ref{step:a2_5}).
Note that in $\tm_2$, we have matched all nodes with indices less than or equal to $1$ as in $F$.
We proceed in the similar fashion by replacing $I$ by $\tm_2$ and obtain $\tm_3$ if $\tm_2 \not= F$.
Note that in each step of the algorithm either we move closer to $F$ or we remain in the previous state itself if the node corresponding to that step is already matched to the same vertex in both $I$ and $F$. Thus the smallest index of node with
non-zero degree in $\tm_k \oplus F$ is a monotone increasing function of $k$.
The path  $I=\tm_1,\tm_2,\ldots,\tm_k=F$ is a canonical path from $I$ to $F$. We show the following.


\begin{theorem}
For any $I, F \in \cN$, Algorithm~\ref{alg:canonicalpath} calculates a unique canonical path say $I = M_1 \rightarrow M_2 \rightarrow \cdots \rightarrow M_k = F$ such that $k = m+1$ and $P_{M_{j}M_{j + 1}} > 0$ $\forall~ j = 1, \cdots, k-1$.
\end{theorem}
\begin{proof}
Given any $I$ and $F$ the canonical path from $I$ to $F$ is constructed using Algorithm \ref{alg:canonicalpath}. 
First note that every intermediate state on the canonical path is a perfect matching. Now, we show that
Algorithm~\ref{alg:canonicalpath} converges in $m$ steps, i.e. the length of the canonical path is $m$.
To see this consider any instant $k$ which corresponds to node $y_k \in V_1$ which is to be resolved in $\tm_k$. Also, let
$\Psi(k) = \{ y \in V_1 : (y,z) \in  \tm_{k} \cap F \}$. We will show that $\Psi(k) \subseteq \Psi(k+1)$. By construction, every $y \in \Psi(k)$ is matched to the same node
in both $\tm_k$ and $F$. Node $y_k$ will fall in one of the cases, (i) $y_k \in \Psi(k)$ or (ii) $y_k \notin \Psi(k)$. In case (i), $\tm_{k+1} = \tm_{k}$ and $\Psi(k+1) = \Psi(k)$. In case (ii), i.e.  $y_k \notin \Psi(k)$, then let $S(k) \subseteq V_2$ denote the set of
vertices that are matched to some $y \in \Psi(k)$ in $F$.
Also, $y_k$ is matched to different nodes in $\tm_k$ and $F$.
Without loss of generality, let $(y_k,z_1) \in \tm_k$ and
$(y_k,z_2) \in F$. Note that $z_1$ and $z_2$ do not belong to $S(k)$.
Now, suppose $z_2$ is not matched in $\tm_k$. Then $\tm_{k+1}$
is obtained by removing $(y_k,z_1)$ and adding $(y_k,z_2)$.
Note that now all nodes $y \leqslant y_k$ are matched in the same way
as that in $F$ under $\tm_{k+1}$. Thus the required follows.
Alternatively, if $z_2$ is matched to some node $y$ in $\tm_k$,
then (1)~$y \notin \Psi(k)$, and (2)~we create $\tm_{k+1}$ by removing
edges $(y_k,z_1)$, $(y,z_2)$ and adding $(y_k,z_2)$, $(y,z_1)$
to $\tm_k$.
In this case also $\Psi(k+1) \supseteq \Psi(k)$.This shows that the length of
the canonical path is $m$.  

Uniqueness of the canonical path is clear from the construction
as we proceed in a specific order. Also, each transition has non-zero 
probability as we change at most two edges in every step. See
Algorithm~\ref{alg:mcmc} for further clarification.
\end{proof}

Construction of the canonical path $I = M_1 \rightarrow M_2 \rightarrow \cdots \rightarrow M_k = F$ 
for a specific example is demonstrated through Figure \ref{fig:canonicalpath} in the Appendix. 

Consider an arbitrary transition $e$ of $\Ma$ that changes the states from $M = (M_{i})$ to $M' = (M'_{i})$. For any transition $e$, let $\cC_{e}$ be the set of ordered pairs $<I, F>$ of perfect matchings, such that $e$ is contained in the canonical path from $I$ to $F$, i.e. $\cC_e = \{ (I, F): e \in I = M_1 \rightarrow M_2 \rightarrow \cdots \rightarrow M_k = F\}$. In order to prove the non-existence of bottlenecks, we construct a map from $\cC_{e}$ to the state space $\cN$. 

\begin{lemma} \label{lem:two}
For every transition $e$, there exists a map $\sigma_{e}:\cN \times \cN \rightarrow \mathcal{N}$ such that $|\cC_{e}| < |\cN|$.
\end{lemma}
\begin{proof}
Given any two states $I$ and $F$, the canonical path construction is defined in such a way that in exactly $m$ transitions we will reach $F$ starting from $I$. Let us say that $e = M \rightarrow M^{\prime}$ is the transition corresponding to the $(k+1)^{\rm th}$ node of set $V_1$. Thus the number of states, say $I^{\prime}$ from which $M$ can be reached by one transition are $(n -(k-1))$, since $(k -1)$ nodes are already matched and their matched pairs are untouched in the $k^{\rm th}$ step. Similarly, the number of ways in which $M$ can be reached by two state transitions are $(n-(k-2))$. Finally, the number of ways in which $M$ can be reached in $k$ steps is $(n - 1)$ ways. Thus the total possible ways in which $I \rightarrow M$ can be reached is $(n -(k-1)) \times (n-(k-2)) \times \cdots \times (n - 1)$.

Now, let us consider the cases where we can reach $F$ from state $M^{\prime}$. The one step cases possible are $(n-(k+1))$, two step possible cases are $(n-(k+2))$ and so on. Finally, the number of possible ways to reach $F$ in $(m - (k+1))$ ways is $(n-m)$. Thus the total number of ways in which $M^\prime \rightarrow F$ can be reached is less than $(n-(k+1)) \times (n-(k+2)) \times \cdots \times (n-m)$.

Thus, the number of $<I, F>$ pairs that have a canonical path from $I$ to $F$ through $e = (M, M^\prime)$
\begin{eqnarray*}
|\mathcal{C}_{e}|  &=& (n - 1)  \cdots  (n-k+2)  (n-k+1)\nonumber \\
&\times & (n-k-1) (n-k-2) \cdots  (n-m) \nonumber \\
|\mathcal{N}| & = & \dfrac{n!}{(n-m)!}  \nonumber \\
&=& n (n-1) \cdots (n-m+2) (n-m+1) \nonumber \\
|\mathcal{C}_{e}| &=& O(n^{m-1})~ {\rm{and}} ~ |\mathcal{N}| = O(n^{m}) \nonumber \\
|\mathcal{C}_{e}| &<& |\mathcal{N}|.
\end{eqnarray*}
 
 Thus given $I, F$ and $e$, let $\sig \in \mathcal{N}$ be a perfect matching. Then, $\sigma_{e}: \mathcal{C}(e) \rightarrow \mathcal{N}$ has $|\mathcal{C}(e)| < |\mathcal{N}|$. This completes the proof of Lemma \ref{lem:two}.
\end{proof}

  The main result of this paper is proving the rapid mixing property of the chain $\Ma$ by finding a lower bound on the conductance of the chain. The concept of canonical paths is used in this. In the following section we prove the rapid mixing property of the Markov chain $\Ma$.

\section{Main Result}\label{sec:mainresult}
In order to prove rapid mixing of $\Ma$ we find a bound on the conductance of the underlying graph associated with the chain. Later, using the conductance bound, we find a bound on the mixing time of the chain $\Ma$. The following results Theorem \ref{th:con} and Theorem \ref{th:time} demonstrate these.
\begin{theorem} \label{th:con}
Consider a bipartite graph $G(V, E)$ with vertex set $V = V_{1} \cup V_{2}$, such that $|V_{1}| = m$,$|V_{2}| = n$ and $m \leqslant n$. The conductance of the Markov chain $\Ma$ that has state space $\mathcal{N}$ of all possible perfect matchings of size $m$, is bounded below by $1/(4\alpha^{3} mn)$, where $\alpha = exp\{\beta(U_{\mathrm{max}} - U_{\mathrm{min}})\}$.
\end{theorem}

The Theorem below gives an upper bound for the mixing time of the Markov chain.
 \begin{theorem} \label{th:time}
The mixing time of the Markov chain $\Ma$ that has state space $\mathcal{N}$ of all possible perfect matchings of size $m$ is bounded by $ \tau_\epsilon \leqslant 32~m^{2}n^{2}\alpha^{6} (-c + m~ {\rm ln}~ n + \mathrm{ln}~ \epsilon^{-1}) $, where $\alpha = exp\{\beta(U_{\mathrm{max}} - U_{\mathrm{min}})\}$ and $c = \mbox{ln}~\alpha$.
 \end{theorem}
 
 The bound on the mixing time is a function of $(U_{\mathrm{max}} - U_{\mathrm{min}})$. By Definition \ref{def:rapid}, the chain $\Ma$ is rapid mixing if the mixing time increases as polynomial in the input size and logarithmic in $(1/\epsilon)$. Thus, for $\Ma$ to be rapid mixing, $(U_{\mathrm{max}} - U_{\mathrm{min}})$ should increase logarithmically with input size or it should be a constant. The instances where $(U_{\mathrm{max}} - U_{\mathrm{min}})$ increase logarithmically with the input size are the feasibility check problems like, job scheduling, graph colouring, multiple knapsack and so on.

Jerrum and Sinclair introduced conditions under which a Markov chain $\Mb$ is rapidly mixing through the following results \cite{JerSin:89}. 
\begin{prop}\label{prop:main}
Let H be an underlying graph of a time-reversible ergodic Markov chain $\Mb$ in which $\displaystyle \min_{i}~P_{i, i} \geqslant \frac{1}{2}$, and let $\pi_{min} = \displaystyle \min_{i}~\pi_{i}$ be the minimum stationary state probability. Then 
\begin{itemize}
\item The total variation distance $\Delta(t)$ of $\Mb$ is bounded by
$
\Delta(t) \leqslant (1 - \Phi(H)^{2}/2)^{t}/\pi_{min}.
$
\item The mixing time of $\Mb$ satisfies $\tau_\epsilon \leqslant 2\Phi^{-2}(\mathrm{ln}~ \pi_{min}^{-1} + \mathrm{ln}~ \epsilon^{-1})$.
\end{itemize}
\end{prop}
The assumption that $\min_{i}~P_{i, i} \geqslant \frac{1}{2}$ ensures that the chain has no negative eigenvalue. Negative eigenvalues corresponds to oscillatory, or ``near periodic'' behaviour and incorporating sufficiently large self loops forbids the occurrence of negative eigenvalues. Proposition \ref{prop:main} allows us to investigate the rapid mixing property of a Markov chain using its underlying graph. Here, rapid mixing is guaranteed if the conductance of the underlying graph is not too small. We prove the following lemma for proving Theorem \ref{th:con}.

\begin{lemma} \label{lem:one}
Consider a bipartite graph $G(V, E)$ with vertex set $V = V_{1} \cup V_{2}$, such that $|V_{1}| = m$,$|V_{2}| = n$ and $m \leqslant n$. Let $I$ and $F$ denote any two states of the Markov chain $\Ma$ whose state space $\mathcal{N}$ is all possible perfect matchings of size $m$ associated with graph $G$. For any $I, F \in \cC(e)$,
\begin{equation*}
\pi_{I}\pi_{F} \leqslant 2mn~a_{\mathrm{max}}^{3}~\pi_{\sig}~w_{e},
\end{equation*}
where  $a_{\mathrm{max}} = exp\{\beta(U_{\mathrm{max}} - U_{\mathrm{min}})\}$.
\end{lemma}
\begin{proof}
The stationary probability of any state $X$ is given by $\pi_{X} = e^{\beta U(X)}/C,$
where $\beta > 0$, $C$ is the partition function and $U(X)$ corresponds to the utility of state $X$. Also, $U_{\mathrm{max}}$ and $U_{\mathrm{min}}$ denotes the maximum and minimum utilities respectively. Then $\pi_{I} = e^{\beta U(I)}/C$,
\begin{eqnarray}
&=& \frac{e^{\beta U(I)}}{C}\left( \frac{e^{\beta U(M)}}{e^{\beta U(M)}} \right ), \nonumber \\
&=& \pi_{M}~e^{\beta(U(I) - U(M))}, \nonumber \\
&\leqslant & \pi_{M}~e^{\beta(U_{\mathrm{max}} - U_{\mathrm{min}})} \label{eq:PI}.
\end{eqnarray}
Consider the path from state $F$ to state $\sig$. Then $ \pi_{F} =e^{\beta U(F)}/C$,
\begin{eqnarray}
&=& \frac{e^{\beta U(F)}}{C}\left( \frac{e^{\beta U(\sig)}}{e^{\beta U(\sig)}} \right ), \nonumber \\
&=& \pi_{\sig}~e^{\beta(U(F) - U(\sig))}, \nonumber \\
&\leqslant & \pi_{\sig}~e^{\beta(U_{\mathrm{max}} - U_{\mathrm{min}})} \label{eq:PF}.
\end{eqnarray} 
For the transitions defined in Section \ref{sec:markov}, the probability of going from any state $i$ to any adjacent state $j$ is given by
\begin{equation}
P_{i, j} = \frac{1}{2mn}a_{ij},
\end{equation}
where $a_{ij} = \mathrm{min}\lbrace exp\{\beta(U(j) - U(i))\}, 1 \rbrace$ is the acceptance probability of the transition from $i$ to $j$. Combining equations \eqref{eq:PI} and \eqref{eq:PF} and later substituting $\pi_{M} = w_{e}/P_{M, M'}$, we get
\begin{eqnarray}
\pi_{I}\pi_{F} &\leqslant & \pi_{M}\pi_{\sig}~e^{2\beta(U_{\mathrm{max}} - U_{\mathrm{min}})}, \nonumber \\
& = & e^{2\beta(U_{\mathrm{max}} - U_{\mathrm{min}})}~\frac{w_{e}}{P_{M, M'}}~\pi_{\sig}. \label{eq:semifinal} 
\end{eqnarray}
Substituting $P_{M, M'} = \dfrac{1}{2mn}~a_{MM'}$ in \eqref{eq:semifinal} we get
\begin{eqnarray}
\pi_{I}\pi_{F} &\leqslant & \left(\frac{2mne^{2\beta(U_{\mathrm{max}} - U_{\mathrm{min}})}}{a_{MM'}}\right)~\pi_{\sig}~w_{e}, \nonumber\\
& \leqslant & \left(\frac{2mn~e^{2\beta(U_{\mathrm{max}} - U_{\mathrm{min}})}}{{a_{\mathrm{min}}}}\right)~\pi_{\sig}~w_{e}, \nonumber\\
& = & 2mn~e^{3\beta(U_{\mathrm{max}} - U_{\mathrm{min}})}~\pi_{\sig}~w_{e}, \label{eq:piIpiF}
\end{eqnarray}
where the minimum and maximum acceptance probabilities are given by $a_{\mathrm{min}} = exp\{\beta(U_{\mathrm{min}} - U_{\mathrm{max}})\}$ and $a_{\mathrm{max}} = exp\{\beta(U_{\mathrm{max}} - U_{\mathrm{min}})\}$ respectively. Thus $\pi_{I}\pi_{F} \leqslant  2mn~a_{\mathrm{max}}^{3}~\pi_{\sig}~w_{e}, $ and this proves Lemma \ref{lem:one}.
\end{proof}

\remove {
The following quantity, referred to as the congestion of an edge, uses a collection of canonical paths
to quantify to what amount that edge is overloaded:
\begin{equation}
\rho_{i,j} = \frac{1}{Q(i,j)} \displaystyle \sum_{ <i, j> \in \mathcal{P} (t)}\pi_{i}\pi_{j}|\gamma_{ij}|,
\end{equation}
where $Q(i, j)$ denotes the capacity of edge $(i,j)$ and the sum represents the total
flow through that edge according to the choice of canonical paths. The congestion of the whole graph
is then defined as $\bar{\rho}:=\mathrm{max}\lbrace\rho(i, j)\rbrace$. Low congestion implies that there are no bottlenecks in
the state space, and the chain can move around fast, which also suggests rapid mixing. 
}

\noindent {\it Proof ~ of ~ Theorem~ \ref{th:con}}:
Let $G$ be the bipartite graph with $m$ number of vertices on one side and $n$ number of vertices on the other side.  Let $H$ be the underlying graph associated with the Markov chain $\Ma(G)$, whose states are perfect matchings of dimension $m$. The conductance of $\Phi(H)$ is defined as
\begin{equation} \label{eq:cond}
\Phi(H) := \displaystyle \min_{{\substack{0 < |S| < \mathcal{|N|} \\ C_{S} \leqslant 1/2}}} \Phi_{S},
\end{equation}
where $\Phi_{S} = F_{S}/C_{S}$ and
\begin{eqnarray*}
C_{S} &:=& \displaystyle{\sum_{M \in S}\pi_{M}}, ~~~~~~~~~~~~~~~~~ \mathrm{the ~ capacity~ of~}S; \\
F_{S} &:=& \displaystyle \sum_{{\substack{M \in S \\ M' \in S^{c} }}}P_{M, M'}\pi_{M},~~~~~~~~~\mathrm{the~ ergodic~ flow~ out~ of~} S. 
\end{eqnarray*}
Thus for any such $S$ the aggregate weight of all paths crossing the cut $S$ to its complement $S^{c}$ in $\mathcal{N}$ is
\begin{equation} \label{eq:cs}
\displaystyle \sum_{I \in S, F \in S^{c}}\pi_{I}\pi_{F} = C_{S}C_{S^{c}} \geqslant \frac{C_{S}}{2}.
\end{equation}
 Consider an arbitrary transition $e$ of $\Ma$ that changes the states from $M = (M_{i})$ to $M' = (M'_{i})$. For any transition $e$, let $\mathcal{C}(e)$ be the set of ordered pairs $<I, F>$ of perfect matchings, such that $e$ is contained in the canonical path from $I$ to $F$. The canonical path from $I$ to $F$ is constructed as per Algorithm \ref{alg:canonicalpath}. For any such transition $e$, there exists a constant $`b$' such that 
 \begin{equation}\label{eq:b}
 \displaystyle \sum_{< I, F > \in \mathcal{C}(e)}\pi_{I}~\pi_{F} \leqslant b~ w_{e},
 \end{equation}
where $w_{e} = \pi_{M}P_{M, M'} = \pi_{M'}P_{M', M}$. Using \eqref{eq:cs} and \eqref{eq:b}, a bound on the ergodic flow out of $S$, where cut$(S)$ denotes the set of transitions crossing the cut from $S$ to $S^{c}$ is
\begin{eqnarray}
F_{S} &=& \displaystyle \sum_{e \in \mathrm{ cut}(S)} w_{e}\nonumber \\
&\geqslant & ~ b^{-1} \displaystyle \sum_{e \in \mathrm{ cut}(S)} \sum_{~~< I, F > \in \mathcal{C}(e)} \pi_{I}\pi_{F} \nonumber \\
&\geqslant & b^{-1} \displaystyle \sum_{I \in S, F \in S^{c}} \pi_{I} \pi_{F} \nonumber \\
&\geqslant & \frac{C_{S}}{2b}.\nonumber
\end{eqnarray}
Using the definition given in equation \eqref{eq:cond}
\begin{equation} \label{eq:condu}
\Phi(H) \geqslant \frac{1}{2b}.
\end{equation} 
 
 Our aim is to define a set of paths so as to find a bound on $b$. In  Lemma \ref{lem:one}, we have defined a set of paths and derived a bound for $\pi_{I}\pi_{F}$. Using \eqref{eq:piIpiF} we get 
\begin{eqnarray}
\displaystyle \sum_{<I, F> \in \mathcal{C}(e)}\pi_{I}\pi_{F} &\leqslant & \displaystyle \sum_{<I, F> \in \mathcal{C}(e)} 2mne^{3\beta(U_{\mathrm{max}} - U_{\mathrm{min}})}~\pi_{\sig}w_{e}, \nonumber \\
& = & 2mne^{3\beta(U_{\mathrm{max}} - U_{\mathrm{min}})}w_{e}\displaystyle \sum_{<I, F> \in \mathcal{C}(e)} \pi_{\sig}, \nonumber  \\
&\leqslant & 2mne^{3\beta(U_{\mathrm{max}} - U_{\mathrm{min}})}w_{e}. \label{eq:inj} 
\end{eqnarray}
Equation \eqref{eq:inj} holds as $|\cC_{e}| < |\cN|$ as proved in Lemma \ref{lem:two}.
Thus the constant $b$ in equation \eqref{eq:condu} is 
\begin{equation}
b =  2mn~exp\{3\beta(U_{\mathrm{max}} - U_{\mathrm{min}})\}.
\end{equation}
Thus the lower bound on the conductance of the underlying graph $H$ of the Markov chain $\Ma$ is derived as
\begin{eqnarray}
\Phi(H) &\geqslant & \frac{1}{4mn~e^{3\beta(U_{\mathrm{max}} - U_{\mathrm{min}})}}, \nonumber \\
 & = & \frac{1}{4mn\alpha^{3}}. \label{eq:condbound}  
\end{eqnarray}
Thus $\Phi(H) \geqslant 1/(4mn\alpha^{3})$ where $\alpha = exp\{\beta(U_{\mathrm{max}} - U_{\mathrm{min}})\}$ and this completes the proof of Theorem \ref{th:con}.  
\EP
 
\noindent {\it Proof ~ of ~ Theorem ~ \ref{th:time}}:
 By Proposition \ref{prop:main}, the mixing time of $\Ma$ is bounded above by
\begin{eqnarray}
\tau_\epsilon &\leqslant & 2\Phi^{-2}(\mathrm{ln}~ \pi_{min}^{-1} + \mathrm{ln}~ \epsilon^{-1}), \nonumber \\
 &\leqslant & 2(4mn\alpha^{3})^{2}(\mathrm{ln}~ \pi_{min}^{-1} + \mathrm{ln}~ \epsilon^{-1}), \nonumber \\
& = & 32~m^{2}n^{2}\alpha^{6}(\mathrm{ln}~ \pi_{min}^{-1} + \mathrm{ln}~ \epsilon^{-1}). \nonumber 
\end{eqnarray}
Using the definition of stationary probability, $\pi_{min} = e^{\beta U_{\mathrm{min}}}/C$. Thus
\begin{eqnarray}
\pi_{min} &\geqslant& \frac{e^{\beta U_{\mathrm{min}}}}{|\cN|e^{\beta U_{\rm max}}}, \nonumber \\
&\geqslant & |\cN |^{-1} e^{\beta (U_{\rm min}- U_{\rm max})}. \nonumber 
\end{eqnarray}
Note that the number of perfect matchings in $G$ is at most $n^{m}$. So the minimum stationary probability $\pi_{\rm min}$ of $\Ma$ satisfies
\begin{eqnarray}
\mathrm{ln}~(\pi_{min}) &\geqslant & \beta(U_{\mathrm{min}} - U_{\mathrm{max}}) - m~ {\rm ln}~ n. \nonumber 
\end{eqnarray}
Thus $ \tau_\epsilon \leqslant 32~m^{2}n^{2}\alpha^{6} \left(\beta(U_{\mathrm{max}} - U_{\mathrm{min}}) + m~{\rm ln}~n + \mathrm{ln}~ \epsilon^{-1}\right) $, where $\alpha = exp\{\beta(U_{\mathrm{max}} - U_{\mathrm{min}})\}$. This completes the proof of Theorem \ref{th:time}.
\EP
 
  The Theorem below proves that the Markov chain $\Ma$ is rapid mixing.  
\begin{theorem}
The Markov chain $\Ma$ is rapid mixing.
\end{theorem}
\begin{proof}
From Theorem \ref{th:time}, we know that $ \tau_\epsilon \leqslant 32~m^{2}n^{2}\alpha^{6} (-\beta(U_{\mathrm{max}} - U_{\mathrm{min}}) + m~{\rm ln}~n + \mathrm{ln}~ \epsilon^{-1}) $, where $\alpha = \exp\{\beta(U_{\mathrm{max}} - U_{\mathrm{min}})\}$. The mixing time is a function in $(U_{\mathrm{max}} - U_{\mathrm{min}})$. Thus by Definition \ref{def:rapid}, the chain $\Ma$ is rapid mixing if $(U_{\mathrm{max}} - U_{\mathrm{min}})$ is a constant or it is increasing logarithmically with input size. However, even when it is increasing polynomially in input size, maximizing $U$ is same as maximizing ln$~U$, since $\rm ln$ is a monotonically increasing function. Thus instead of optimizing for $U$, we can always do optimization with respect to ${\rm ln}~U$. Thus $(U_{\mathrm{max}} - U_{\mathrm{min}})$ can be made a function that increase logarithmically in input size. Thus the chain $\Ma$ is rapid mixing.
\end{proof} 
\section{Conclusion}\label{sec:conclu}
The matching problem considered in this paper is a complete bipartite graph matching problem with the objective of maximizing a global utility function. Conventional matching algorithms cannot be used for solving this because of the global nature of the utility function. We proved the NP-hardness of the problem by showing a reduction of the well known 3-SAT problem. Thus there are no known polynomial time algorithm for solving this. To this end, we constructed a Markov chain whose every state is a perfect matching and the state space is all possible perfect matchings. The transitions in the chain are modelled in such a way that in every transition the chain favours a higher utility state over other adjacent states. The chain is a DTMC which is aperiodic, irreducible and time reversible. The rapid mixing of the chain is proved by finding a lower bound on the conductance of the underlying graph using a canonical path construction. 

\bibliographystyle{abbrv}
\bibliography{myreferences}

\appendix
\section{Construction of Canonical Path}
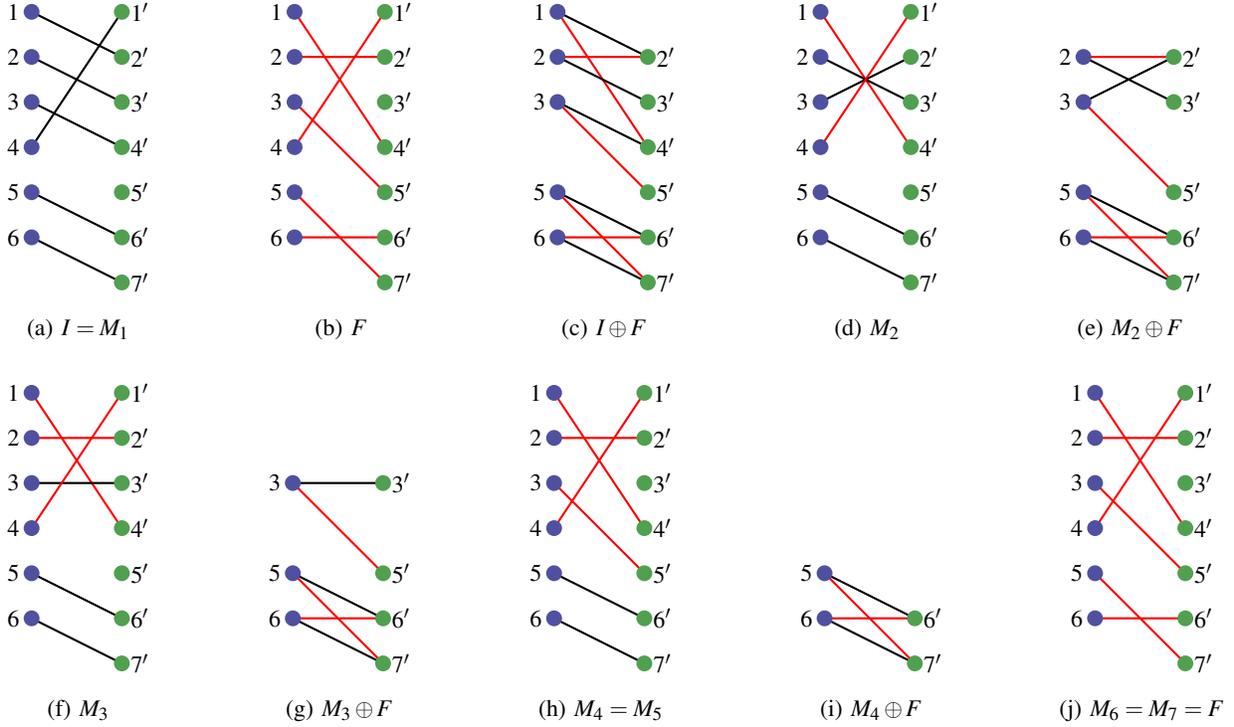
\begin{figure*}[t]
\definecolor{blue}{RGB}{80,80,160}
\definecolor{green}{RGB}{80,160,80} 
    \begin{subfigure}[b]{0.2 \textwidth}
       \centering
\begin{tikzpicture} [scale = 0.3]
       \draw [thick] (-2, 0)   --   (2,-2);
       \draw [thick] (-2, -2)  --   (2, -4);
       \draw [thick] (-2, -4)  --   (2,-6);
       \draw [thick] (-2, -6)   --   (2,0);
       \draw [thick] (-2, -8)  --   (2, -10);
       \draw [thick] (-2, -10)  --   (2,-12);      
          
          \node at (-2.8,0) {\small $1$};
          \node at (-2.8,-2.0) {\small $2$};
          \node at (-2.8,-4.0) {\small $3$};
          \node at (-2.8,-6) {\small $4$};
          \node at (-2.8,-8) {\small $5$};
          \node at (-2.8,-10) {\small $6$};
          \fill[blue] (-2,0) circle (10.0 pt);
          \fill[blue] (-2,-2) circle (10.0 pt);
          \fill[blue] (-2,-4) circle (10.0 pt);
          \fill[blue] (-2,-6) circle (10.0 pt);
          \fill[blue] (-2,-8) circle (10.0 pt);
          \fill[blue] (-2,-10) circle (10.0 pt);
          \node at (2.8,0) {\small $1'$};
          \node at (2.8,-2.0) {\small $2'$};
          \node at (2.8,-4.0) {\small $3'$};
          \node at (2.8,-6) {\small $4'$};
          \node at (2.8,-8) {\small $5'$};
          \node at (2.8,-10) {\small $6'$};
          \node at (2.8,-12) {\small $7'$};
          \fill[green] (2,0) circle (10.0 pt); 
          \fill[green] (2,-2) circle (10.0 pt);
          \fill[green] (2,-4) circle (10.0 pt);
          \fill[green] (2,-6) circle (10.0 pt); 
          \fill[green] (2,-8) circle (10.0 pt);
          \fill[green] (2,-10) circle (10.0 pt);
          \fill[green] (2,-12) circle (10.0 pt);
      \end{tikzpicture}
        \caption{$I = M_1$}
        \label{fig:I}
    \end{subfigure} ~\hspace{-5 mm}
    \begin{subfigure}[b]{0.2 \textwidth}
       \centering
       \begin{tikzpicture} [scale = 0.3]
       \draw [thick, red] (-2, 0)   --   (2,-6);
       \draw [thick, red] (-2, -2)  --   (2, -2);
       \draw [thick, red] (-2, -4)  --   (2,-8);      
       \draw [thick, red] (-2, -6)  --   (2,0)  ;
       \draw [thick, red] (-2, -8)  --   (2,-12);
       \draw [thick, red] (-2, -10) --   (2,-10);      
          
          \node at (-2.8,0) {\small $1$};
          \node at (-2.8,-2.0) {\small $2$};
          \node at (-2.8,-4.0) {\small $3$};
          \node at (-2.8,-6) {\small $4$};
          \node at (-2.8,-8) {\small $5$};
          \node at (-2.8,-10) {\small $6$};
          \fill[blue] (-2,0) circle (10.0 pt);
          \fill[blue] (-2,-2) circle (10.0 pt);
           \fill[blue] (-2,-4) circle (10.0 pt);
          \fill[blue] (-2,-6) circle (10.0 pt);
          \fill[blue] (-2,-8) circle (10.0 pt);
           \fill[blue] (-2,-10) circle (10.0 pt);
           
          \node at (2.8,0) {\small $1'$};
          \node at (2.8,-2.0) {\small $2'$};
          \node at (2.8,-4.0) {\small $3'$};
          \node at (2.8,-6) {\small $4'$};
          \node at (2.8,-8) {\small $5'$};
          \node at (2.8,-10) {\small $6'$};
          \node at (2.8,-12) {\small $7'$};
          \fill[green] (2,0) circle (10.0 pt); 
          \fill[green] (2,-2) circle (10.0 pt);
          \fill[green] (2,-4) circle (10.0 pt);
          \fill[green] (2,-6) circle (10.0 pt); 
          \fill[green] (2,-8) circle (10.0 pt);
          \fill[green] (2,-10) circle (10.0 pt);
          \fill[green] (2,-12) circle (10.0 pt);
      \end{tikzpicture}
        \caption{$F$}
        \label{fig:F}
    \end{subfigure} ~\hspace{-5 mm}
    \begin{subfigure}[b]{0.2 \textwidth}
    \centering
       \begin{tikzpicture} [scale = 0.3]
       \draw [thick] (-2, 0)   --   (2,-2);
       \draw [thick, red] (-2, 0)  --   (2, -6);
       \draw [thick, red] (-2, -2)  --   (2,-2);
       \draw [thick] (-2, -2)  --   (2,-4);     
       \draw [thick] (-2, -4)  --   (2,-6)  ;
       \draw [thick, red] (-2, -4)  --   (2,-8);
       \draw [thick, red] (-2, -8) --   (2,-12);   
       \draw [thick] (-2, -8) --   (2,-10); 
       \draw [thick, red] (-2, -10) --   (2,-10); 
       \draw [thick] (-2, -10) --   (2,-12);   
          
          \node at (-2.8,0) {\small $1$};
          \node at (-2.8,-2.0) {\small $2$};
          \node at (-2.8,-4.0) {\small $3$};
          \node at (-2.8,-8) {\small $5$};
          \node at (-2.8,-10) {\small $6$};
          \fill[blue] (-2,0) circle (10.0 pt);
          \fill[blue] (-2,-2) circle (10.0 pt);
           \fill[blue] (-2,-4) circle (10.0 pt);
          \fill[blue] (-2,-8) circle (10.0 pt);
          \fill[blue] (-2,-10) circle (10.0 pt);
          \node at (2.8,-2.0) {\small $2'$};
          \node at (2.8,-4.0) {\small $3'$};
          \node at (2.8,-6) {\small $4'$};
          \node at (2.8,-8) {\small $5'$};
          \node at (2.8,-10) {\small $6'$};
          \node at (2.8,-12) {\small $7'$};
          \fill[green] (2,-2) circle (10.0 pt);
          \fill[green] (2,-4) circle (10.0 pt);
          \fill[green] (2,-6) circle (10.0 pt); 
          \fill[green] (2,-8) circle (10.0 pt);
          \fill[green] (2,-10) circle (10.0 pt);
          \fill[green] (2,-12) circle (10.0 pt);
      \end{tikzpicture}
        \caption{$I \oplus F$}
        \label{fig:IF}
    \end{subfigure} ~\hspace{-5 mm}
       \begin{subfigure}[b]{0.2 \textwidth}
       \centering
       \begin{tikzpicture} [scale = 0.3]
       \draw [thick, red] (-2, 0)   --   (2,-6);
       \draw [thick] (-2, -2)  --   (2, -4);
       \draw [thick] (-2, -4)  --   (2,-2);      
       \draw [thick, red] (-2, -6)  --   (2,0)  ;
       \draw [thick] (-2, -8)  --   (2,-10);
       \draw [thick] (-2, -10) --   (2,-12);      
          
          \node at (-2.8,0) {\small $1$};
          \node at (-2.8,-2.0) {\small $2$};
          \node at (-2.8,-4.0) {\small $3$};
          \node at (-2.8,-6) {\small $4$};
          \node at (-2.8,-8) {\small $5$};
          \node at (-2.8,-10) {\small $6$};
          \fill[blue] (-2,0) circle (10.0 pt);
          \fill[blue] (-2,-2) circle (10.0 pt);
           \fill[blue] (-2,-4) circle (10.0 pt);
          \fill[blue] (-2,-6) circle (10.0 pt);
          \fill[blue] (-2,-8) circle (10.0 pt);
          \fill[blue] (-2,-10) circle (10.0 pt);
          \node at (2.8,0) {\small $1'$};
          \node at (2.8,-2.0) {\small $2'$};
          \node at (2.8,-4.0) {\small $3'$};
          \node at (2.8,-6) {\small $4'$};
          \node at (2.8,-8) {\small $5'$};
          \node at (2.8,-10) {\small $6'$};
          \node at (2.8,-12) {\small $7'$};
          \fill[green] (2,0) circle (10.0 pt); 
          \fill[green] (2,-2) circle (10.0 pt);
          \fill[green] (2,-4) circle (10.0 pt);
          \fill[green] (2,-6) circle (10.0 pt); 
          \fill[green] (2,-8) circle (10.0 pt);
          \fill[green] (2,-10) circle (10.0 pt);
          \fill[green] (2,-12) circle (10.0 pt);
      \end{tikzpicture}
        \caption{$ M_2$}
        \label{fig:M2}
    \end{subfigure} ~\hspace{-5 mm}
   \begin{subfigure}[b]{0.2 \textwidth}
    \centering
      \begin{tikzpicture} [scale = 0.3]
       \draw [thick, red] (-2, -2)  --   (2,-2);
       \draw [thick] (-2, -2)  --   (2,-4);     
       \draw [thick] (-2, -4)  --   (2,-2)  ;
       \draw [thick, red] (-2, -4)  --   (2,-8);
       \draw [thick, red] (-2, -8) --   (2,-12);   
       \draw [thick] (-2, -8) --   (2,-10); 
       \draw [thick, red] (-2, -10) --   (2,-10); 
       \draw [thick] (-2, -10) --   (2,-12);   
          
          \node at (-2.8,-2.0) {\small $2$};
          \node at (-2.8,-4.0) {\small $3$};
          \node at (-2.8,-8) {\small $5$};
          \node at (-2.8,-10) {\small $6$};
          \fill[blue] (-2,-2) circle (10.0 pt);
           \fill[blue] (-2,-4) circle (10.0 pt);
          \fill[blue] (-2,-8) circle (10.0 pt);
          \fill[blue] (-2,-10) circle (10.0 pt);
          \node at (2.8,-2.0) {\small $2'$};
          \node at (2.8,-4.0) {\small $3'$};
          \node at (2.8,-8) {\small $5'$};
          \node at (2.8,-10) {\small $6'$};
          \node at (2.8,-12) {\small $7'$};
          \fill[green] (2,-2) circle (10.0 pt);
          \fill[green] (2,-4) circle (10.0 pt);
          \fill[green] (2,-8) circle (10.0 pt);
          \fill[green] (2,-10) circle (10.0 pt);
          \fill[green] (2,-12) circle (10.0 pt);
      \end{tikzpicture}
        \caption{$M_2 \oplus F$}
        \label{fig:M2F}
    \end{subfigure} ~\hspace{-5 mm}
    \newline
    
   \begin{subfigure}[b]{0.2\textwidth}
    \centering
     \begin{tikzpicture} [scale = 0.3]    
       \draw [thick, red] (-2, 0)   --   (2,-6);
       \draw [thick, red] (-2, -2)  --   (2, -2);
       \draw [thick] (-2, -4)  --   (2,-4);      
       \draw [thick, red] (-2, -6)  --   (2,0)  ;
       \draw [thick] (-2, -8)  --   (2,-10);
       \draw [thick] (-2, -10) --   (2,-12);      
          
          \node at (-2.8,0) {\small $1$};
          \node at (-2.8,-2.0) {\small $2$};
          \node at (-2.8,-4.0) {\small $3$};
          \node at (-2.8,-6) {\small $4$};
          \node at (-2.8,-8) {\small $5$};
          \node at (-2.8,-10) {\small $6$};
          \fill[blue] (-2,0) circle (10.0 pt);
          \fill[blue] (-2,-2) circle (10.0 pt);
           \fill[blue] (-2,-4) circle (10.0 pt);
          \fill[blue] (-2,-6) circle (10.0 pt);
          \fill[blue] (-2,-8) circle (10.0 pt);
           \fill[blue] (-2,-10) circle (10.0 pt);
           
          \node at (2.8,0) {\small $1'$};
          \node at (2.8,-2.0) {\small $2'$};
          \node at (2.8,-4.0) {\small $3'$};
          \node at (2.8,-6) {\small $4'$};
          \node at (2.8,-8) {\small $5'$};
          \node at (2.8,-10) {\small $6'$};
          \node at (2.8,-12) {\small $7'$};
          \fill[green] (2,0) circle (10.0 pt); 
          \fill[green] (2,-2) circle (10.0 pt);
          \fill[green] (2,-4) circle (10.0 pt);
          \fill[green] (2,-6) circle (10.0 pt); 
          \fill[green] (2,-8) circle (10.0 pt);
          \fill[green] (2,-10) circle (10.0 pt);
          \fill[green] (2,-12) circle (10.0 pt);
           \end{tikzpicture}
        \caption{$M_3$}
        \label{fig:M3}
    \end{subfigure}~\hspace{-4 mm}
   \begin{subfigure}[b]{0.2 \textwidth}
    \centering
      \begin{tikzpicture} [scale = 0.3]
       \draw [thick] (-2, -4)  --   (2,-4)  ;
       \draw [thick, red] (-2, -4)  --   (2,-8);
       \draw [thick, red] (-2, -8) --   (2,-12);   
       \draw [thick] (-2, -8) --   (2,-10); 
       \draw [thick, red] (-2, -10) --   (2,-10); 
       \draw [thick] (-2, -10) --   (2,-12);   
          
          \node at (-2.8,-4.0) {\small $3$};
          \node at (-2.8,-8) {\small $5$};
          \node at (-2.8,-10) {\small $6$};
           \fill[blue] (-2,-4) circle (10.0 pt);
          \fill[blue] (-2,-8) circle (10.0 pt);
          \fill[blue] (-2,-10) circle (10.0 pt);
          \node at (2.8,-4.0) {\small $3'$};
          \node at (2.8,-8) {\small $5'$};
          \node at (2.8,-10) {\small $6'$};
          \node at (2.8,-12) {\small $7'$};
          \fill[green] (2,-4) circle (10.0 pt);
          \fill[green] (2,-8) circle (10.0 pt);
          \fill[green] (2,-10) circle (10.0 pt);
          \fill[green] (2,-12) circle (10.0 pt);
      \end{tikzpicture}
        \caption{$M_3 \oplus F$}
        \label{fig:M3F}
    \end{subfigure}~\hspace{-4 mm}
       \begin{subfigure}[b]{0.2\textwidth}
        \centering
     \begin{tikzpicture} [scale = 0.3]    
       \draw [thick, red] (-2, 0)   --   (2,-6);
       \draw [thick, red] (-2, -2)  --   (2, -2);
       \draw [thick, red] (-2, -4)  --   (2,-8);      
       \draw [thick, red] (-2, -6)  --   (2,0)  ;
       \draw [thick] (-2, -8)  --   (2,-10);
       \draw [thick] (-2, -10) --   (2,-12);      
          
          \node at (-2.8,0) {\small $1$};
          \node at (-2.8,-2.0) {\small $2$};
          \node at (-2.8,-4.0) {\small $3$};
          \node at (-2.8,-6) {\small $4$};
          \node at (-2.8,-8) {\small $5$};
          \node at (-2.8,-10) {\small $6$};
          \fill[blue] (-2,0) circle (10.0 pt);
          \fill[blue] (-2,-2) circle (10.0 pt);
           \fill[blue] (-2,-4) circle (10.0 pt);
          \fill[blue] (-2,-6) circle (10.0 pt);
          \fill[blue] (-2,-8) circle (10.0 pt);
           \fill[blue] (-2,-10) circle (10.0 pt);
           
          \node at (2.8,0) {\small $1'$};
          \node at (2.8,-2.0) {\small $2'$};
          \node at (2.8,-4.0) {\small $3'$};
          \node at (2.8,-6) {\small $4'$};
          \node at (2.8,-8) {\small $5'$};
          \node at (2.8,-10) {\small $6'$};
          \node at (2.8,-12) {\small $7'$};
          \fill[green] (2,0) circle (10.0 pt); 
          \fill[green] (2,-2) circle (10.0 pt);
          \fill[green] (2,-4) circle (10.0 pt);
          \fill[green] (2,-6) circle (10.0 pt); 
          \fill[green] (2,-8) circle (10.0 pt);
          \fill[green] (2,-10) circle (10.0 pt);
          \fill[green] (2,-12) circle (10.0 pt);
           \end{tikzpicture}
        \caption{$M_4= M_5$}
        \label{fig:M4}
    \end{subfigure} ~\hspace{-4 mm}
       \begin{subfigure}[b]{0.2 \textwidth}
        \centering
      \begin{tikzpicture} [scale = 0.3]
       \draw [thick, red] (-2, -8) --   (2,-12);   
       \draw [thick] (-2, -8) --   (2,-10); 
       \draw [thick, red] (-2, -10) --   (2,-10); 
       \draw [thick] (-2, -10) --   (2,-12);   
          
          \node at (-2.8,-8) {\small $5$};
          \node at (-2.8,-10) {\small $6$};
          \fill[blue] (-2,-8) circle (10.0 pt);
          \fill[blue] (-2,-10) circle (10.0 pt);
          \node at (2.8,-10) {\small $6'$};
          \node at (2.8,-12) {\small $7'$};
          \fill[green] (2,-10) circle (10.0 pt);
          \fill[green] (2,-12) circle (10.0 pt);
      \end{tikzpicture}
        \caption{$M_4 \oplus F$}
        \label{fig:M4F}
    \end{subfigure} ~\hspace{-4 mm}
        \begin{subfigure}[b]{0.2 \textwidth}
       \centering
       \begin{tikzpicture} [scale = 0.3]
       \draw [thick, red] (-2, 0)   --   (2,-6);
       \draw [thick, red] (-2, -2)  --   (2, -2);
       \draw [thick, red] (-2, -4)  --   (2,-8);      
       \draw [thick, red] (-2, -6)  --   (2,0)  ;
       \draw [thick, red] (-2, -8)  --   (2,-12);
       \draw [thick, red] (-2, -10) --   (2,-10);      
          
          \node at (-2.8,0) {\small $1$};
          \node at (-2.8,-2.0) {\small $2$};
          \node at (-2.8,-4.0) {\small $3$};
          \node at (-2.8,-6) {\small $4$};
          \node at (-2.8,-8) {\small $5$};
          \node at (-2.8,-10) {\small $6$};
          \fill[blue] (-2,0) circle (10.0 pt);
          \fill[blue] (-2,-2) circle (10.0 pt);
           \fill[blue] (-2,-4) circle (10.0 pt);
          \fill[blue] (-2,-6) circle (10.0 pt);
          \fill[blue] (-2,-8) circle (10.0 pt);
           \fill[blue] (-2,-10) circle (10.0 pt);
           
          \node at (2.8,0) {\small $1'$};
          \node at (2.8,-2.0) {\small $2'$};
          \node at (2.8,-4.0) {\small $3'$};
          \node at (2.8,-6) {\small $4'$};
          \node at (2.8,-8) {\small $5'$};
          \node at (2.8,-10) {\small $6'$};
          \node at (2.8,-12) {\small $7'$};
          \fill[green] (2,0) circle (10.0 pt); 
          \fill[green] (2,-2) circle (10.0 pt);
          \fill[green] (2,-4) circle (10.0 pt);
          \fill[green] (2,-6) circle (10.0 pt); 
          \fill[green] (2,-8) circle (10.0 pt);
          \fill[green] (2,-10) circle (10.0 pt);
          \fill[green] (2,-12) circle (10.0 pt);
      \end{tikzpicture}
        \caption{$M_6 = M_7 = F$}
        \label{fig:M5}
        \end{subfigure}~\hspace{-5 mm}
    \caption{Constructing the canonical path $I = M_1 \rightarrow M_2 \rightarrow M_3 \rightarrow M_4  \rightarrow M_5  \rightarrow M_6 \rightarrow M_7 = F$.} \label{fig:canonicalpath}
\end{figure*}
Figure \ref{fig:canonicalpath} can be explained in the following way. We start with the first node of the set $V_1$. In the above example it is clear that $V_1 = \{1,2,3,4,5,6 \}$. Select the lowest numbered node, i.e. $1$. The edges corresponding to node $1$ in $I \oplus F$ are $(1,2')$ and $(1,4')$, belonging to $I$ and $F$ respectively. Pick the edge $(1, 4')$ and perform the transition in $I$ to get $M_2$ as given in Figure \ref{fig:M2}. Note that $|M_2 \oplus F|< |I \oplus F|$. Now, the lowest degree node in $M_2 \oplus F$ is $2$ and the corresponding edge to be selected for transition is $(2,2')$. Perform the transition to get $M_3$ as shown in Figure \ref{fig:M3} and find $M_3 \oplus F$. Again, $|M_3 \oplus F|< |M_2 \oplus F|$. The next step is corresponding to node $3$ and the edge that should be used for transition is $(3, 5')$, which gives $M_4$ as shown in Figure \ref{fig:M4}. Since node $4$ is connected to $1$ in both $I$ and $F$, the next state $M_5 = M_4$. In the next stage, we select node $5$ and corresponding edge $(5,7')$ for transition. At the end of this transition we obtain matching $M_6$ which is indeed the desired matching $F$. This is because,  node $6$ is already resolved in the previous stage transition and thus $M_7 = M_6 = F$. 

\section{Some Proofs} 
\noindent {\it Proof of Lemma \ref{lem:DTMC}}:
Consider the discrete time stochastic process $\{X_t\}_{t\geqslant 0}$ where $X_t$ is a random variable that denotes the matching used by Algorithm \ref{alg:mcmc} in $t^{\rm th}$ iteration. The sequence $X_{0}, X_{1},\cdots $ is a random sequence as the chain $\Ma$ progress through the state space $\cN$. From the transition definition it is clear that the state at the $(t+1)^{\rm th}$ instant, $X_{t+1}$, depends only on the previous state $X_t$. That is, $P~(X_{t+1}|X_{0}, X_{1}, \cdots, X_{t}) = P~(X_{t+1}|X_{t})$. Thus $\Ma$ is a Discrete Time Markov Chain (DTMC) on $\cN$. 

In order to prove $\Ma$ is irreducible, we need to show that any arbitrary state $F$ can be reached from any arbitrary state $I$. From the unique construction of the canonical path (see Algorithm \ref{alg:canonicalpath}) it is clear that given any $I$, there exists a unique path of exactly $m$ transitions to reach state $F$. Let $\gamma_{IF}$ is the canonical path from $I$ to $F$. Then, $length~(\gamma_{IF}) = m$ $\forall~ I, F$. Thus $\Ma$ is irreducible.

Now we will prove that $\Ma$ is aperiodic. Let $M_{1}$ and $M_{2}$ be two adjacent states of $\Ma$, the acceptance of a transition from $M_{1}$ to $M_{2}$ is defined based on a utility comparison between the respective states. At $t^{\rm th}$ instant when state $X_{t} = M_{1}$ the transition probabilities to state $M_{2}$ is modelled as
\begin{eqnarray*} 
\small 
X_{t + 1} =
\begin{cases}
M_{2}~~~ \mbox{if}~ U(M_{2}) \geqslant U(M_{1}), \\
M_{2}~~~ \mathrm{with~ prob.}~ \exp\{\beta(U(M_2) - U(M_1))\}, \mbox{ otherwise }.
\end{cases}
\end{eqnarray*}
Thus when  $U(M_{2}) \geqslant U(M_{1})$, the state $X_{t+1}$ is $M_2$ with probability 1. However, when $ U(M_{2}) < U(M_{1})$, then transition from $M_1$ to $M_2$ is with certain probability. The probability associated with the transition introduces self loops making the chain $\Ma$ aperiodic. Thus the chain $\Ma$ is a DTMC which is irreducible and aperiodic. 
\EP

\noindent {\it  Proof of Lemma \ref{lem:reversible}}:
The DTMC $\cMb$  is said to be reversible if for any two adjacent states $M_1$ and $M_2$ in $\cN$, $\pi_{M_1}(\beta)P_{M_1M_2}(\beta) = \pi_{M_2}(\beta)P_{M_2M_1}(\beta)$. Every entry $P_{ij}(\beta)$ of the probability transition matrix $P(\beta)$ is given by
\[ P_{ij}(\beta) = \frac{1}{2mn}a_{ij}(\beta), \]
where $m$ and $n$ denotes the number of nodes in the respective sets of the bipartite graph and $a_{ij}(\beta)$ is the acceptance probability from state $i$ to state $j$. We know $a_{ij}(\beta) = \mbox{min}~\{1, \exp\{\beta(U(j) - U(i))\}\}$. Thus
\begin{eqnarray*} 
a_{M_1M_2}(\beta) =
\begin{cases}
1~~~~~~~~~~~~~~~~~~~~~~~~~~~~~~~~ \mbox{if}~ U(M_{2}) \geqslant U(M_{1}), \\
exp\{\beta(U(M_2) - U(M_1))\}~~~~ \mbox{if}~ U(M_{2}) < U(M_{1}).
\end{cases}
\end{eqnarray*}
Now consider two cases: (a) $U(M_2) \geqslant U(M_1)$ and (b) $U(M_2) < U(M_1)$. When $U(M_2) \geqslant U(M_1)$, 
\begin{eqnarray}
\pi_{M_1}(\beta)P_{M_1M_2}(\beta) &=& \frac{\exp\{\beta U(M_1)\}}{\sum_{M^\prime\in \mathcal{N}}\exp\{\beta U(M^\prime)\}} \times \frac{1}{2|E|} \times 1, \nonumber \\
\pi_{M_2}(\beta)P_{M_2M_1}(\beta) &=& \frac{\exp\{\beta U(M_2)\}}{\sum_{M^\prime\in \mathcal{N}}\exp\{\beta U(M^\prime)\}} \times \frac{1}{2|E|} \times \frac{\exp\{\beta U(M_1) \}}{\exp\{\beta U(M_2) \}}  \nonumber \\
&=& \frac{\exp\{\beta U(M_1)\}}{\sum_{M^\prime\in \mathcal{N}}\exp\{\beta U(M^\prime)\}} \times \frac{1}{2|E|}, \nonumber \\
&=& \pi_{M_1}(\beta)P_{M_1M_2}(\beta). \nonumber
\end{eqnarray}

When $U(M_2) < U(M_1)$, 
\begin{eqnarray}
\pi_{M_1}(\beta)P_{M_1M_2}(\beta)  &=& \frac{\exp\{\beta U(M_1)\}}{\sum_{M^\prime\in \mathcal{N}}\exp\{\beta U(M^\prime)\}} \times \frac{1}{2|E|}\times \frac{\exp\{\beta U(M_2) \}}{\exp\{\beta U(M_1) \}} , \nonumber \\
 &=&   \frac{\exp\{\beta U(M_2)\}}{\sum_{M^\prime\in \mathcal{N}}\exp\{\beta U(M^\prime)\}}\times \frac{1}{2|E|}, \nonumber \\
\pi_{M_2}(\beta)P_{M_2M_1}(\beta)   &=&   \frac{\exp\{\beta U(M_2)\}}{\sum_{M^\prime\in \mathcal{N}}\exp\{\beta U(M^\prime)\}} \times \frac{1}{2|E|} \times 1  \nonumber \\
 &=&  \pi_{M_1}(\beta)P_{M_1M_2}(\beta). \nonumber
\end{eqnarray}

\remove{Using similar arguments we can show that $\pi_{M_1}(\beta)P_{M_1M_2} (\beta) = \pi_{M_2}(\beta)P_{M_2M_1}(\beta)$ for case (b) as well.} 
\noindent
Thus DTMC $\cMb$ is time reversible and this completes the proof of
Lemma \ref{lem:reversible}.
\EP

\end{document}